\RequirePackage{rotating}

\documentclass[a4paper,12pt]{amsart}
\usepackage{graphicx}
\usepackage[]{color}
\usepackage[]{tikz} \usetikzlibrary{matrix, chains, arrows}

\usepackage[margin=3cm]{geometry}

\usepackage{rotating}
\newtheorem{thm}{Theorem}[section]

\newtheorem{prop}[thm]{Proposition}
\theoremstyle{definition}
\newtheorem{defn}[thm]{Definition}
\theoremstyle{remark}

\numberwithin{equation}{section}

\newcommand{\signedpermutation}{
\pgfmatharray{\sp}{0}\let\n\pgfmathresult
\pgfmathparse{360.0/\n}\let\segment\pgfmathresult
\pgfmathparse{\segment/2}\let\shift\pgfmathresult
\def\radius{1.5cm}
\def\labelrad{1.7cm}
\def\regionboundaryin{1.4cm}
\def\regionboundaryout{1.6cm}
\scalebox{1.3}{\begin{tikzpicture}
\foreach \x in {1,2,...,\n}
{
  \draw[thick] (360-\x*\segment+90:\regionboundaryin)
             --(360-\x*\segment+90:\regionboundaryout);
  \pgfmatharray{\sp}{\x}\let\tmp\pgfmathresult
  \pgfmathparse{int(\tmp)}\let\tmp\pgfmathresult
  \node at (360-\x*\segment+90+\shift:\labelrad) {\tmp};
  \pgfmathparse{360-(\x-1)*\segment+90}\let\alpha\pgfmathresult;
  \pgfmathparse{360-(\x-1)*\segment+90-\segment}\let\beta\pgfmathresult;
  \pgfmathgreater{\tmp}{0}\let\decision\pgfmathresult
  \ifnum \decision=1
    \draw[color=black,very thick,>=latex,->] (\alpha:\radius) arc (\alpha:\beta:\radius);
  \else
    \draw[color=black,very thick,>=latex,->] (\beta:\radius) arc (\beta:\alpha:\radius);
  \fi
};
\end{tikzpicture}}}
\newcommand{\unsignedpermutation}{
\pgfmatharray{\sp}{0}\let\n\pgfmathresult
\pgfmathparse{360.0/\n}\let\segment\pgfmathresult
\pgfmathparse{\segment/2}\let\shift\pgfmathresult
\def\radius{1.5cm}
\def\labelrad{1.7cm}
\def\regionboundaryin{1.4cm}
\def\regionboundaryout{1.6cm}
\scalebox{1.3}{\begin{tikzpicture}
\foreach \x in {1,2,...,\n}
{
  \draw[thick] (360-\x*\segment+90:\regionboundaryin)
             --(360-\x*\segment+90:\regionboundaryout);
  \pgfmatharray{\sp}{\x}\let\tmp\pgfmathresult
  \pgfmathparse{int(\tmp)}\let\tmp\pgfmathresult
  \node at (360-\x*\segment+90+\shift:\labelrad) {\tmp};
  \pgfmathparse{360-(\x-1)*\segment+90}\let\alpha\pgfmathresult;
  \pgfmathparse{360-(\x-1)*\segment+90-\segment}\let\beta\pgfmathresult;
  \pgfmathgreater{\tmp}{0}\let\decision\pgfmathresult
  \ifnum \decision=1
    \draw[color=black,very thick] (\alpha:\radius) arc (\alpha:\beta:\radius);
  \else
    \draw[color=gray,very thick] (\beta:\radius) arc (\beta:\alpha:\radius);
  \fi
};
\end{tikzpicture}}}

\newcommand{\Q}{\mathbb{Q}}

\title{A mean first passage time genome rearrangement distance}%
\author{Andrew R. Francis${}^\ast$}%
\address{Centre for Research in Mathematics and Data Science, Western Sydney University, Australia}%
\email{a.francis@westernsydney.edu.au}%
\thanks{${}^\ast$ Corresponding author: \texttt{a.francis@westernsydney.edu.au}}
\author{Henry P. Wynn}
\address{Centre for the Analysis of Time Series, London School of Economics, UK}%
\email{h.wynn@lse.ac.uk}%

\thanks{ARF acknowledges the support of the Australian Research Council via Discovery Project DP180102215.}%

\date{\today}%

\begin{document}

\begin{abstract}
This paper introduces a new way to define a genome rearrangement distance, using the concept of mean first passage time from probability theory.  Crucially, this distance estimate provides a genuine metric on genome space.  We develop the theory and introduce a link to a graph-based zeta function.  The approach is very general and can be applied to a wide variety of group-theoretic models of genome evolution.
\end{abstract}
\maketitle

\section{Introduction}

Estimating evolutionary distances between organisms is a key ingredient in most approaches to phylogenetic reconstruction.  Making such estimates broadly involves two steps: specifying an evolutionary model (the way in which the organisms can change), and deciding what metric to use.  The most straightforward approach to the latter is to ask for the least number of steps between the two organisms under the model; this is the \emph{minimal} distance.

By definition, a minimal distance can only \emph{under}estimate the true distance, and there is considerable interest in finding ways to estimate distance that are less problematic~\cite{jukes1969evolution,wang2001estimating,serdoz2017maximum}.  In this paper, we take a new approach to this problem by adapting a construction from probability theory, within a framework that also exploits group theory.  Thus the methods in this paper bring together mathematical tools from disparate fields to address a problem in molecular evolution.  Specifically, we show how one may be able to calculate the \emph{mean first passage time (MFPT)} between two organisms, and we put this forward as an alternative to the minimal distance.  

The ``true distance'' is the actual number of inversions that take place on the path between two genomes that have the same gene content. 
As said, the minimal distance is an underestimate of this, and unlikely to be exhibited by a  random process. The distance in this paper can be considered as an alternative to the maximum likelihood estimate approach~\cite{egrinagy2013group,francis2014algebraic}, and has some potential advantages over it. First, the MLE does not always exist in the sense that the likelihood function may not have a unique  maximum, or indeed any maximum at all. Second, the MLE has a infinite series computation, needing truncation, whereas the formulae in this paper are closed form. Third, looking at the {\em first} passage time of a path from genome $A$ to genome $B$, avoids the issue of multiple visits of paths that are considered in the MLE approach.  

The mean first passage time to a target  is the average time it takes for a random process to reach the target for the first time, along some un-prespecified path.  This will depend typically on the structure of the underlying graph that represents an evolutionary process on genome space; in our case a group structure. The first passage time is a well-studied quantity in the theory of Markov processes and network flows.  It can be defined on any strongly connected directed graph (that is, for which each pair of vertices is connected by a directed path).  A moment generating function for the first passage time can be computed using the structure of the graph, including closed loops within the graph. This is an approach due to Mason in the 1950s in the system theory literature~\cite{mason1953feedback}.  It can also be computed using determinants of matrices associated with the adjacency matrix of the graph~\cite{butler1997stochastic}. 

We apply the latter approach, using the adjacency matrix of the Cayley graph of a group, $G$, that represents genome space under a particular model of evolution.  The Cayley graph of a group is a graph whose vertices represent the group elements, and each directed edge corresponds to multiplication by one of the generators of the group, with the reverse arrow being multiplication by the inverse element (here we use right multiplication). Note that typically a list of generators is not unique, in which case nor is the Cayley graph or our metric.  The method  builds on the algebraic models of bacterial inversion introduced in~\cite{egrinagy2013group}.  In this framework, each vertex of the Cayley graph  corresponds to a unique genome arrangement and the edges to possible evolutionary events.  While we will consider the general situation, in most examples arising from this framework the generators of the group are self-inverse (for instance an inversion done twice returns the genome to the original state), which means that the edges of the corresponding Cayley graph are often undirected.

Now, assume each (directed) edge $(i,j)$ has a Markov transition probability $p_{ij}$ from $i$ to $j$.  In addition, there is an independent  random  ``passage time''  $X_{ij}$, from $i$ to $j$. The \emph{first} passage time is then the time taken for a random walk starting at $i$  to reach the target state $j$, {\em for the first time}.  We assume that there is at least one path between any two vertices. Once the  probability distribution function of this first passage time is given, our distance $d_{i,j}$ is defined as its mean $\mu_{i,j}$. In summary we assume:
\begin{enumerate}
  \item The moment generating functions, $\{m_{ij}(s)\}$ of the passage times  $\{X_{ij} \}$ are known, and
  \item The Markov transition probabilities $\{p_{ij} \}$ are known.
\end{enumerate}

In the terminology of Markov processes, the Cayley graphs may not yield an aperiodic system, and therefore the underlying Markov chain will not, be positive recurrent.  Thus, that part of the stochastic process theory which requires aperiodicity will not apply.  But the formulae here are quite general and can also be written in terms of purely combinatorial zeta functions for graphs and sub-graphs associated with paths (described in Section~\ref{a:mason.loops}).

The structure of the paper is as follows.  We begin in Section~\ref{s:bio.background} with an introduction to the motivating problem from bacterial genome rearrangements, together with an explanation of the group-theoretic framework with which we study it.  This is followed by a  straightforward account of the Markov flow theory in Section~\ref{s:mason}.  The definition of our mean first passage time distance is given in Section~\ref{s:mfpt.distance}, and this is followed by some fully worked examples in Section~\ref{s:cayley}. We return to the Mason rule in Section~\ref{a:mason.loops} giving links to zeta functions on graphs. A discussion section points to further research.

\section{Background to bacterial genome rearrangements and group-theoretic models}\label{s:bio.background}

In this section we explain the basic biological information about bacterial genome rearrangements, and we present some necessary details of the algebraic models from~\cite{egrinagy2013group} that this paper relies on.

{Large scale rearrangements}, in which whole regions of the chromosome are moved around relative to each other, are a significant driver of evolutionary adaptation in the case of bacteria.   Large scale rearrangements are uncommon in eukaryotic nuclear DNA, though they are a feature of mitochondrial DNA, probably because of its heritage as an ancestral bacterial invasion (almost all bacterial genomes and mitochondrial DNA is circular).   The mechanisms within the cell that give rise to these rearrangements revolve around enzymes called \emph{site-specific recombinases}, which cut two double-helical strands that are adjacent in the cell, and rejoin them in a new way, changing the sequence of the chromosome.  They facilitate the movement of genes around a chromosome (which can have a fitness effect) as well as the acquisition of new genetic material (through {horizontal gene transfer}), and deletion of redundant DNA.

We will focus on single-celled organisms and large scale rearrangements on a single chromosome because they play an important role in establishing phylogenetic relationships in these cases.  This set-up includes all bacteria, but omits models such as the influential double-cut-and-join (DCJ)~\cite{yancopoulos2005efficient,bergeron2006unifying}.  The rearrangement events we will focus on are inversion and translocation, because both of these events are ``invertible'' (can be undone), and so are suited to a group-theoretic treatment~\cite{egrinagy2013group,francis2014algebraic}.  By a \emph{model of rearrangement}, we mean three things: a genome structure; a set of allowable operations that rearrange the regions on the genome; and a probability distribution on the operations.  This slightly generalizes the algebraic structure described in~\cite{egrinagy2013group}, in the inclusion of the probability distribution.  We will always assume that both genomes have the same set of regions in their genomes. 

The two models of general interest we consider and in which we include or omit the orientation of the DNA, are illustrated in Figure~\ref{f:genomes} for the circular case.  The other genome geometry we will mention is a \emph{lineal} chromosome, which has regions arranged along a line, again either including orientation or not (we use the word lineal to avoid confusion from use of the word ``linear'' for this geometry).

\begin{figure}[ht]
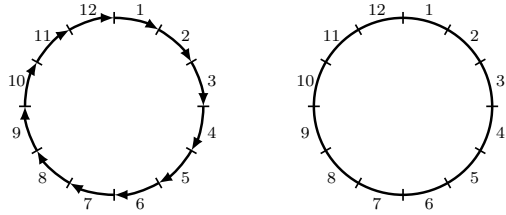

\tiny
\resizebox{7cm}{!}{%
\def\sp{{12,1,2,3,4,5,6,7,8,9,10,11,12}}\signedpermutation \hspace{1cm}\unsignedpermutation
}
\caption{Model genomes on 12 regions, with orientation and without. The regions are simply segments of DNA that have been preserved in the set of genomes under study, and may represent genes or segments containing several genes.}\label{f:genomes}
\end{figure}

\emph{Inversion} takes a region or set of contiguous regions and reverses their relative positions, while  \emph{translocation} takes a region or set of contigous regions and moves it past another set of regions, as illustrated in Figure~\ref{f:inversions}.  Such rearrangements are called \emph{signed} if the orientation of the regions is considered.

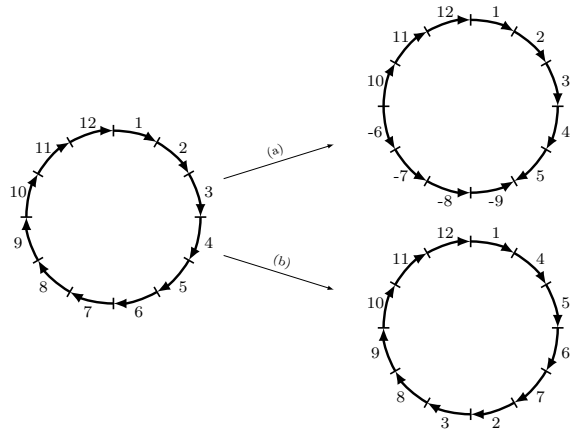
\begin{figure}[ht]
\tiny
\resizebox{8cm}{!}{%
\begin{tikzpicture}
\node (1) {\def\sp{{12,1,2,3,4,5,6,7,8,9,10,11,12}}\signedpermutation};
\begin{scope}[xshift=8cm,yshift=2.5cm]
\node (2) {\def\sp{{12,1,2,3,4,5,-9,-8,-7,-6,10,11,12}}\signedpermutation};
\end{scope}
\draw[->,>=latex](1)--(2) node[pos=0.5,above,sloped] {(a)};
\begin{scope}[xshift=8cm,yshift=-2.5cm]
\node (3) {\def\sp{{12,1,4,5,6,7,2,3,8,9,10,11,12}}\signedpermutation};
\end{scope}
\draw[->,>=latex](1)--(3) node[pos=0.5,above,sloped] {(b)};
\end{tikzpicture}
}
\caption{(a) A signed inversion of positions 6 to 9. (b) A translocation of the regions in positions 2 and 3 past position 7.}\label{f:inversions} 
\end{figure}

The dominant method for establishing distance in any of these models is by calculating the \emph{minimal} distance.  In some cases, this can be done very fast.  In particular, there is a large body of literature establishing the minimal distance under the model in which all inversions occur with equal probability~\cite{sankoff1992gene,bafna1993genome}.  These methods are chiefly combinatorial, constructing a ``breakpoint graph'' whose features give the distance.  Recently, group-theoretic methods have been introduced that rely on results in Coxeter groups, and are effective for inversions of only two adjacent regions (we will call these \emph{2-inversions})~\cite{egrinagy2013group,francis2014algebraic}.

Several alternatives to the minimal distance to estimate the evolutionary distance in rearrangement models have been proposed.  These are summarized in~\cite{serdoz2017maximum}, which develops a way to calculate a maximum likelihood estimate for the evolutionary distance, further developed in~\cite{sumner2017representation}.  Aside from the MLE of~\cite{serdoz2017maximum}, these are dominated by approaches that use the relationships between the minimal distance and the true distance obtained using simulation studies (for instance~\cite{wang2001estimating,dalevi2008expected}). 

In this paper we focus on rearrangement models that are suited to a group-theoretic approach, namely those for which the permitted evolutionary operations are ``invertible'': there is a permitted operation that undoes each one.  Usually in this context such operations are self-inverse: they undo themselves.  
The core of this approach is the observation that if a family of rearrangement events are each invertible, then, because the composition of such events is associative ($a\circ(b\circ c)=(a\circ b)\circ c)$), they generate a group that is acting on the set of genome arrangements.  A sequence of such rearrangement events is then a composition of a sequence of group generators.  
This gives an interpretation of the distance, and of evolutionary history, in terms of paths and path lengths on the Cayley graph of the group~\cite{clark2016bacterial}, which we now describe.  

If one particular genome arrangement is chosen as a ``reference genome'', then every other genome arrangement can be associated with a unique group element defined by the product of the generators on a path to it from the reference genome.  Despite there being possibly many paths to the genome, the product is well-defined regardless of which path is chosen, giving a single group element.  That is, considering genome rearrangements as generators of a group, and choosing a reference genome, defines a one-to-one correspondence between the set of genomes and the set of elements of the group generated by the rearrangements.  

The correspondence between genome arrangements and group elements means that the space of genomes can now be considered as a Cayley graph: the graph whose vertices are elements of the group, and in which there is an edge from group element $g$ to $h$ if $h=gs$ for some generator $s$.  The minimal distance between two genomes is then simply the length of a minimal path on the Cayley graph, and the true evolutionary history is a walk on the Cayley graph.

The methods we develop here are applicable to several models of rearrangement.  First, the widely studied cases of lineal or circular genomes in which regions of DNA are considered as either oriented (signed) or unoriented.  Second, one may allow different evolutionary events, including inversions of different lengths of DNA (measured by the number of regions inverted in a single event), as well as translocations.
Some of the models that are relevant are listed in Table~\ref{tab:models}, in Appendix~\ref{s:evol.models}.

\section{Markov flow models}\label{s:mason}

Consider a directed simple graph $G(V,E)$ with vertex set $V$, $|V|=n$ and edge set $E$ consisting of  ordered pairs $(i,j)$.   By ``simple'', we mean that $G$ has no parallel edges (so that $(i,j)$ defines at most one edge) and no loops $(i,i)$ (edges from a vertex to itself). It helps intuition to consider transition as the movement of a hypothetical particle.

We assume that given any vertices $i$ and $j$, a particle starting at vertex $i$ can reach vertex $j$ along a directed path (this strongly connected condition is satisfied for Cayley graphs, for instance).
If $(i,j) \in E$ and the particle is at vertex $i$, the particle travels directly to vertex $j$ with probability $p_{ij}>0$, and
$\sum_{j \neq i} p_{ij} = 1$ for all $i$; so the particle having reached vertex $i$ must move immediately away from $i$. Let the random variable $X_{ij}$ denote the {\em inter-arrival time} (passage time) along edge $(i,j) \in E$.  We assume that all the $X_{ij}$ for the travel of the particle are independent and that each $X_{ij}$ has a well-defined moment generating function:
$$m_{ij}(s) = {\mathbb E}_{X_{ij}}\{\exp(s X_{ij})\}.$$

Let $Y_{ij}$ be the \emph{first passage time} from vertex $i$ to vertex $j$. To study this we consider the directed subgraph
$G(V,E_{[j]} )$ where
$$E_{[j]} = E \setminus \{ (j,k) : (j,k) \in E\}.$$
That is, we remove all edges out of $j$, thereby turning  $j$ into an absorbing state. The first passage time $Y_{ij}$ is the sum of all  $X_{ij}$ realised from vertex $i$ till the particle arrives at  vertex $j$ for the first time, for the graph $G(V,E_{[j]})$. Note that the particle can spend an arbitrary large but finite length of time in circuits (if there are circuits) and different visits to an edge, say $(i,j)$, are assumed to give independent copies of $X_{ij}$.

What we call here the \emph{Mason rule} (Theorem~\ref{t:mason}) is a version of the original rule which is also available via Markov renewal theory (see~\cite{butler1997stochastic,pyke1961markov,howard1960dynamic,howard1971risk}).
It  is sometimes referred to as the {\em cofactor rule}. The  rule  gives the moment generating function, $\tilde{M}_{ij}(s)$ of the $Y_{ij}$ in terms of the $p_{ij}$ and $m_{ij}(s)$.  

Let $P= \{p_{ij}\}$ denote the Markov matrix for the process, and let $Q = \{p_{ij} m_{ij}(s) \}$ denote the {\em transmittance} matrix $Q(s)$, both of size $|V|\times|V|$, whose entries we will write $q_{ij}(s):=p_{ij} m_{ij}(s)$.  Note that $q_{ij}(s) = 0 $ if $p_{ij} = 0$.  The following theorem is a version of Mason's rule. 

\begin{thm}[Mason's rule~\cite{mason1953feedback}]\label{t:mason}
For a  Markov flow model with transmittance matrix $Q$, the moment generating function $\tilde{M}_{ij}(s)$ of the first passage time $Y_{ij}$ from vertex $i$ to vertex $j$, is given by the ${i,j}$ entry in the matrix $(I - Q_{[j]})^{-1}(s)$, where $Q_{[j]}$ is  obtained from $Q(s)$ by setting all transmittances $q_{jk},\;(j,k) \in E$ equal to zero. 
\end{thm}

Note that the matrix inverse in the theorem exists via the theory of absorbing states for Markov chains.

\section{The mean first passage time distance}\label{s:mfpt.distance}

We define our distance  $d_{ij}$, $ i \rightarrow j$ as follows.

\begin{defn}
For a directed graph $G(V,E)$ with Markov transition matrix $P=\{ p_{ij}\}$ and inter-arrival moment generating function matrix $M(s) = \{m_{ij}(s)\}$, we define the mean first passage time (MFPT) distance as
$$d_{ij} = {\mathbb E}_{Y_{ij}} (Y_{ij}),$$
where $Y_{ij}$ is the first passage time from vertex $i$ to vertex $j$.
\end{defn}

The distance depends on $M(s)$ only via the edge means $\mu_{ij} = \mbox{E} (X_{ij}) =  m'_{ij}(0)$.
Let $e_i$ be the $i$-th unit vector (1 in entry $i$ and 0 elsewhere).
Then, using the $e_i$ to pick out the entries, Mason's rule (Theorem~\ref{t:mason}) immediately gives us:
\begin{prop}\label{p:dij.formula}
\[
d_{ij}  =  \frac{\partial}{\partial s} \left( e_i^T \left(I - Q_{[j]}(s)\right)^{-1} e_j\right) \big\vert_{s = 0}.
\]
\end{prop}

While this is already a practically useful expression, we can further develop an explicit expression for $d_{ij}$ as follows.
By the formula for differentiating inverses,
\begin{eqnarray*}
d_{ij} & = & \frac{\partial}{\partial s} \left( e_i^T \left(I - Q_{[j]}(s)\right)^{-1} e_j\right) \big\vert_{s = 0}\\
         & =  &-\left\{e_i^T \left(I- Q_{[j]}(s)\right)^{-1} \left(\frac{\partial}{\partial s} \left( I - Q_{[j]}(s) \right)\right) \left(I - Q_{[j]}(s)\right)^{-1}e_j \right\} \Big\vert_{s=0}.\notag
\end{eqnarray*}

Note that
$$I-Q_{[j]}(s) = I - P _{[j]} \circ M_{[j]}(s),$$
where  ``$ \circ $'' means the Schur (Hadamard, entry-wise) product, and $P_{[j]}$ is the matrix $P$ with the $j$-th row set to zero.

Define $M = M'(0) =  \{\mu_{ij}\} $ to be the matrix of  passage time edge means and write $M_{[j]}=M'_{[j]}(0)$ for the matrix formed from $M$  by excluding means for edges out of vertex $j$  (so that the $j$-th row is zero). Then, using \ref{t:mason} we have an explicit formula for the distances
(noting that $m_{ij}(0) = 1$ for all $(i,j) \in E$ with $i \neq j$):
$$
d_{ij} = e_i (I - P_{[j]})^{-1}(P_{[j]} \circ M_{[j]})(I-  P_{[j]})^{-1} e_j.
$$
Examples of this computation can be found in Section~\ref{s:cayley}, exploiting the efficiency of the linear algebra version in Proposition~\ref{p:dij.formula}.

Note that the $(r,t)$ entry of $(I-P_{[j]}))^{-1}$ is the passage time from $r$ to $t$ for the graph in which $j$ has been made an absorbing state and in which every edge has fixed unit travel time. Let us label these entries $\{z_{rt}\}$. Then expanding the matrices we have
$$
d_{ij} = \sum_{ (r,t)} z_{ir} p_{rt}\mu_{rt} z_{tj}.
 $$
Each term $\mu_{rt}$ in this expansion has coefficient
$$\theta_{r,t\;:\;i,j} = z_{ir}z_{tj} p_{rt},$$
whose informal interpretation is as follows. 
Each distinct path from $i$ to $j$ may or may not use the edge $(r,t)$. Those
that use $(r,t)$ each contribute a weight to the mean $\mu_{rt}$. Each such path must enter at $r$
and leave at $t$. The paths into $r$ have total passage time $z_{ir}$ and those out of $t$ and being absorbed at $j$ have total passage time $z_{tj}$. Independence gives the  contribution as $z_{ir}z_{tj}p_{rt}$. A technical point is that because the matrix $P_{[j]}$ has an absorbing state, and the full chain is connected, all the terms are finite.

As defined, the distance $d_{ij}$ satisfies the directional triangular inequality: $d_{ik} \leq d_{ij} + d_{jk}$ for distinct vertices $i,j,k$. This is proved by splitting events into two disjoint types: (i) starting in $i$ we reach $j$ first before $k$ and (ii) starting in $i$ we reach $k$ first before $j$. In both cases the inequality holds:  in (i) we have equality and in (ii) by assumption $d_{ik} \leq d_{ij}$. The symmetry needed to be an ordinary distance requires $d_{ij} = d_{ji}$ for all $i \neq j$. The Cayley graph version defined next satisfies this condition.

\section{The Cayley graph case}\label{s:cayley}

As explained in the introduction, the distance $d_{ij}$ between group elements $g_i$ and $g_j$ is defined to be the mean of the first passage time, when the full directed graph is the Cayley graph, $C(G)$ of the group generated by the elements corresponding to the evolutionary events of interest.

Although the theory allows a general moment generating function $m_{ij}(s)$, for simplicity 
in the Cayley graph case we set all the $m_{ij}$ equal and all the $p_{ij}$ equal.  If there are $k$  generators for the group,
then for any edge $(i,j)$ of $C(G)$,
$$q_{ij}(s) = \frac{1}{k}m(s),$$
and
$$Q(s) = \frac{1}{k}m(s)A(G),$$
where $A(G)$ is the adjacency matrix of the Cayley graph of $G$: $C(G)$. Then we have a simple version:
\begin{equation}\label{simple}
I - Q_{[j]}(s)  = I - z A_{[j]} 
\end{equation}
with $z= \frac{1}{k}m(s)$ and $A_{[j]}$ is the adjacency matrix of the graph obtained from
$C(G)$, by deleting all the arrows out of  $j$. 

We begin with an elementary result that describes when two group elements have the same mean first passage time (considered as distances from the identity element), contingent on properties related to the \emph{normaliser} of the set of generators $S$ of the group $G$.  Recall that the normaliser of $S$ in $G$ is the subgroup of $G$ defined by  $N_G(S):=\{\sigma\in G\mid \sigma^{-1}S\sigma=S\}$.  The normaliser acts on the group $G$ by conjugation, and the orbits of this action partition both the set of generators, by definition of the normaliser, and the group itself, by definition of an orbit.

\begin{prop}\label{p:normaliser}
Suppose we have a random walk on a Cayley graph $C(G)$ beginning at the identity element, in which the generators $S$ of $G$ are all equally probable.  Suppose in addition that the inter-arrival times $X_{ij}$ have the same distribution for all edges $(i,j)$ corresponding to generators that are in the same orbit of the action of $N_{G}(S)$. Then, if two group elements are in the same orbit of $N_G(S)$, they will have the same mean first passage time.
\end{prop}

\begin{proof}
Theorem 4.2 of~\cite{clark2016bacterial} shows that two group elements conjugate by the normaliser of the generators will have order-isomorphic ``intervals'', meaning that the path structures from the identity to each of them are isomorphic as partially ordered sets.  The additional requirements here about the probabilities of the generators, and the inter-arrival times, ensures that the distribution of first passage times to each element will be the same, and in particular the mean first passage times will be the same.
\end{proof}
A key property of the Cayley graph is {\em vertex transitivity}, which, heuristically, means that the graph looks the same from any vertex. Thus, any edge $ g \rightarrow gs$, where $s$ is a generator can be transformed to $e \rightarrow s$ by left multiplying by $g^{-1}$.  We see, more generally, that  the whole graph remains the same if we premultiply by a fixed $g^{-1}$ at every vertex.

For simplicity, in the examples below we will use $m_{ij}(s) = e^s$, for all $(i,j) \in E$, which corresponds to a fixed (non-random) time between vertices of one unit. 
That is, we take $X_{ij}$ to be a discrete random variable which takes the value 1 with probability 1, so that the moment generating function is just $e^s$.

\subsection{Example: $S_3$ with standard (Coxeter) generators}

We carry out the computations for the Cayley graph of the symmetric (permutation) group on three elements.
With identity $e$ and two generators $g_1$ and $g_2$, respectively (transpositions $(1\ 2)$ and $(2\ 3)$ as shown in Figure~\ref{f:S3.circ}, acting on the right), the elements, labeled as vertices are
\begin{align*}\{&v_1=e,v_2= g_1=(1\ 2), v_3= g_2=(2\ 3),\\ &v_4 =g_1g_2=(1\ 2\ 3),v_5 = g_2g_1=(1\ 3\ 2), 
v_6 = g_1g_2g_1= g_2g_1g_2=(1\ 3)\}.\end{align*}
\begin{figure}[ht]
\includegraphics[height=4cm]{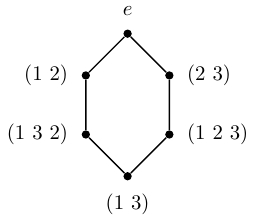}\hspace{1cm}
\includegraphics[height=4cm]{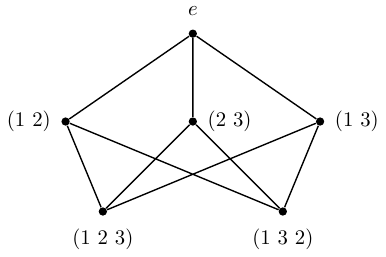} 
\caption{Cayley graphs of $S_3$ with standard generators $\{(1\ 2),(2\ 3)\}$ and circular generators $\{(1\ 2),(2\ 3),(1\ 3)\}$.}
\label{f:S3.circ}
\end{figure}

The adjacency matrix and its absorbing form based on $v_6$ are
$$ A = \left(
         \begin{array}{cccccc}
           0 & 1 & 1 & 0 & 0 & 0 \\
           1 & 0 & 0 & 1 & 0 & 0 \\
           1 & 0 & 0 & 0 & 1 & 0 \\
           0 & 1 & 0 & 0 & 0 & 1 \\
           0 & 0 & 1 & 0 & 0 & 1 \\
           0 & 0 & 0 & 1 & 1 & 0 \\
         \end{array}
       \right),\;\;
A_{[6]} = \left(
         \begin{array}{cccccc}
           0 & 1 & 1 & 0 & 0 & 0 \\
           1 & 0 & 0 & 1 & 0 & 0 \\
           1 & 0 & 0 & 0 & 1 & 0 \\
           0 & 1 & 0 & 0 & 0 & 1 \\
           0 & 0 & 1 & 0 & 0 & 1 \\
           0 & 0 & 0 & 0 & 0 & 0 \\
         \end{array}
       \right).
$$
The expressions below give the entries of the matrix $(I - z A_{[6]})^{-1}$.
\begin{align*}
((I - z A_{[2]})^{-1})_{1,2} = ((I - z A_{[6]})^{-1})_{4,6} & = \frac{z-2z^3}{1-3z^2}\\
((I - z A_{[4]})^{-1})_{1,4} = ((I - z A_{[6]})^{-1})_{2,6} & = \frac{z^2}{1-3z^2}\\
((I - z A_{[6]})^{-1})_{1,6} & = \frac{2z^3}{1-3z^2}.
\end{align*}
Note that by vertex transitivity, we do not need to calculate $A_{[2]}$ or $A_{[4]}$: the set of paths from vertex 1 to vertex 4 ($e$ to $g_1g_2$) is in one-to-one correspondence with the set of paths from vertex 2 to vertex 6 ($g_1$ to $g_1g_2g_1$). 

We now use the simple form (\ref{simple}), setting $k=2$, $z=\frac{1}{2}e^s$,  differentiating with respect to $s$ and setting $s=0$. This gives the distinct mean first passage times   $e$ to (i) $v_2$ or $v_3$,  (ii) $v_4$ or $v_5$ and (iii) $v_6$, as, respectively,  5, 8 and 9.  This compares to the shortest (geodesic) distances of 1,  2 and 3 respectively.

That some distances are equal  supports the fact that in the group theoretic sense elements which are conjugate by the normaliser of the generators will have identical mean first passage time, as in Proposition~\ref{p:normaliser}.

\subsection{Example: $S_3$ with circular generators}

Working with circular generators adds a generator across the top of the appropriate circle diagram (representing the circular genome with three regions), namely $(1\ 3)$, so that the generators are $\{(1\ 2),(2\ 3),(1\ 3)\}$ (note that we use cycle notation). The Cayley graph is a complete bipartite graph, namely $K_{3,3}$ (Figure~\ref{f:S3.circ}).

Ordering the elements $e,(12),(23),(13),(123),(132)$, the adjacency matrix with the last state absorbing is
\[ A_{[6]} = \left(
         \begin{array}{cccccc}
           0 & 1 & 1 & 1 & 0 & 0 \\
           1 & 0 & 0 & 0 & 1 & 1 \\
           1 & 0 & 0 & 0 & 1 & 1 \\
           1 & 0 & 0 & 0 & 1 & 1 \\
           0 & 1 & 1 & 1 & 0 & 0 \\
           0 & 0 & 0 & 0 & 0 & 0 \\
         \end{array}
       \right).
\]
The entries in the last column of $(I-zA_{[6]})^{-1}$ are $3z^2/(1-6z^2),
z/(1-6z^2),
z/(1-6z^2),
z/(1-6z^2),
3z^2/(1-6z^2),
1$ respectively.
By substituting $k=3, z=\frac{1}{3}e^s$, differentiating with respect to $s$ and then setting $s=0$, we obtain mean first passage times $(6,5,5,5,6,0)$ (which compare to $(2,1,1,1,2,0)$ for the shortest distance). This means the MFPT distance between any two elements on opposite sides of the bipartite graph is 5, and between two distinct elements on the same side is 6.

\subsection{Example: $S_4$ with standard and circular generators}

The Cayley graphs for $S_4$, for both standard and circular generating sets are shown in Figure~\ref{f:cayley.S4}, Appendix~\ref{s:cayley.figs}.

Instead of calculating the inverse of the matrix $(I-zA_{[24]})$ (which is computationally difficult), we calculate the terms we need by using the fact that any distance on the Cayley graph is equivalent to a distance to the longest word (see for instance~\cite[Part I]{humphreys1990reflection}).  So we only need to compute that final column of $(I-zA_{[24]})^{-1}$, which can be done by a simple row reduction.  That is to say, if the final column is given by the vector $v$ then it is the solution to the matrix equation $(I-zA_{[24]})v=e_{24}$, (recalling that $e_{24}\in\Q^{24}$ is the vector whose entries are zero except for a 1 in the last position).

Substituting $z=\frac{1}{3} e^s$, differentiating with respect to $s$, and evaluating at $s=0$, we have the distances to the longest word shown in Table~\ref{t:fpt.distances.S4.standard}.

\begin{table}
\begin{tabular}{ccc||ccc}
\multicolumn{3}{c}{lineal generators} & \multicolumn{3}{c}{Circular generators} \\
\hline
 group       & minimal   & MFPT 		& group        & minimal   & MFPT   	\\
 element     & distance  & distance 	&    element   & distance  & distance\\
\hline	
$()$           & 6 &  1296/28 = 46.3	&  $()$           & 4 &  32 \\
$(1\ 2)$       & 5 &  1273/28 = 45.5	&  $(1\ 2)$       & 3 &  31 \\
$(3\ 4)$       & 5 &  1273/28 = 45.5	&  $(2\ 3)$       & 3 &  31 \\
$(2\ 3)$       & 5 &  1258/28 = 44.9	&  $(1\ 4)$       & 3 &  31 \\
$(1\ 2)(3\ 4)$ & 4 &  1242/28 = 44.4	&  $(1\ 3)$       & 3 &  31 \\
$(1\ 3\ 2)$    & 4 &  1197/28 = 42.8	&  $(3\ 4)$       & 3 &  31 \\
$(1\ 2\ 3)$    & 4 &  1197/28 = 42.8	&  $(2\ 4)$       & 3 &  31 \\
$(2\ 4\ 3)$    & 4 &  1197/28 = 42.8	&  $(1\ 4\ 3\ 2)$ & 3 &  31 \\
$(2\ 3\ 4)$    & 4 &  1197/28 = 42.8	&  $(1\ 2\ 3\ 4)$ & 3 &  31 \\
$(1\ 3)$       & 3 &  1153/28 = 41.2	&  $(1\ 3\ 2)$    & 2 &  30 \\
$(2\ 4)$       & 3 &  1153/28 = 41.2	&  $(1\ 2\ 3)$    & 2 &  30 \\
$(1\ 3\ 4\ 2)$ & 3 &  1096/28 = 39.1	&  $(2\ 4\ 3)$    & 2 &  30 \\
$(1\ 2\ 4\ 3)$ & 3 &  1096/28 = 39.1	&  $(1\ 4\ 3 )$   & 2 &  30 \\
$(1\ 2\ 3\ 4)$ & 3 &  1081/28 = 38.6	&  $(2\ 3\ 4)$    & 2 &  30 \\
$(1\ 4\ 3\ 2)$ & 3 &  1081/28 = 38.6	&  $(1\ 4\ 2)$    & 2 &  30 \\
$(1\ 4\ 3)$    & 2 &  981 /28 = 35.0	&  $(1\ 3\ 4)$    & 2 &  30 \\
$(1\ 4\ 2)$    & 2 &  981 /28 = 35.0	&  $(1\ 2\ 4)$    & 2 &  30 \\
$(1\ 3\ 4)$    & 2 &  981 /28 = 35.0	&  $(1\ 2)(3\ 4)$ & 2 &  28 \\
$(1\ 2\ 4)$    & 2 &  981 /28 = 35.0	&  $(1\ 4)(2\ 3)$ & 2 &  28 \\
$(1\ 3)(2\ 4)$ & 2 &  810 /28 = 28.9	&  $(1\ 2\ 4\ 3)$ & 1 &  23 \\
$(1\ 4)$       & 1 &  682 /28 = 24.4	&  $(1\ 3\ 4\ 2)$ & 1 &  23 \\
$(1\ 4\ 2\ 3)$ & 1 &  625 /28 = 22.3	&  $(1\ 4\ 2\ 3)$ & 1 &  23 \\
$(1\ 3\ 2\ 4)$ & 1 &  625 /28 = 22.3	&  $(1\ 3\ 2\ 4)$ & 1 &  23 \\
$(1\ 4)(2\ 3)$ & 0 &  0 				&  $(1\ 3)(2\ 4)$ & 0 &  0  \\
\hline
\end{tabular}
\caption{First passage time distances in $S_4$ to the longest word, using Coxeter generators (lineal genome) on the left ($\{(1\ 2),(2\ 3),(3\ 4)\}$) and circular generators on the right ($\{(1\ 2),(2\ 3),(3\ 4),(4\ 1)\}$), in comparison to the minimal distances in the middle columns. Ordering by minimal length in each case. }\label{t:fpt.distances.S4.standard.length.ordering}\label{t:fpt.distances.S4.standard}
\end{table}

In Table~\ref{t:fpt.distances.S4.standard.length.ordering}, one can  also observe the additional symmetry in the generating set provided by the circular generators (corresponding to inversions on a circular genome as in~\cite{egrinagy2013group}), as follows.  
Recalling Proposition~\ref{p:normaliser}, observe that with the lineal generators
$\{(1\ 2),(2\ 3),(3\ 4)\},$
the normaliser of these is trivial,  namely $\{e\}$. This is in contrast to the case of  circular generators, 
$\{(1\ 2),(2\ 3),(3\ 4),(4\ 1)\},$
for which the normaliser is
$$\{e,(1\ 3),(2\ 4),(1\ 3)(2\ 4),(1\ 2\ 3\ 4),(1\ 4\ 3\ 2),(1\ 4)(2\ 3),(1\ 2)(3\ 4)\}.$$  
More elements are conjugate by elements of the (larger) normalizer in the circular case and so we see in Table~\ref{t:fpt.distances.S4.standard.length.ordering} that there are fewer distinct mean first passage times using the circular generators than the lineal.  

\subsection{Abelian groups}\label{s:abelian}
For some general classes of groups we are able to obtain exact combinatorial formulae for the distance. Here we give the result
for the  abelian group of order $2^k$ with $k$  generators $a_1, \ldots, a_k$ and each element of order two:
$a_1^2 = a_2^2 = \cdots  = a_k^2=e$, where $e$ is the identity, and an  example to show the method (the full proof will appear in a separate paper).

Consider the case $n = 4$. The elements of the Cayley graph can be divided into 5 sets, according to the lengths of the elements in terms of the generators: 
$$
\begin{array}{c}
L_0=\{e\}, L_1=\{a_1,a_2, a_3,a_4\}, L_2= \{a_1a_2, a_1a_3, a_1a_4, a_2 a_3, a_2 a_4, a_3a_4 \} , \\
L_3=\{ a_1a_2a_3, a_1a_2a_4, a_1a_3a_4, a_2a_3a_4\},  L_4= \{ a_1a_2 a_3a_4\}.
\end{array}
$$
To understand the structure of the Cayley graph it is convenient to work inductively, doubling the size of the matrix every time we add a new generator.   Thus for $n=4$  the rows and columns are ordered as follows:
$$
\begin{array}{c}
e, a_1 \\
a_2,a_1a_2 \\ 
a_3,a_1a_3,a_2a_3,a_1a_2a_3\\
a_4,a_1a_4,a_2a_4,a_1a_2a_4,a_3a_4,a_1a_3a_4,a_2a_3a_4,a_1a_2a_3a_4.
\end{array}
$$

The adjacency matrix is then:
$$A = 
\left(
\begin{array}{cccccccccccccccc}
0 & 1 & 1 & 0 & 1 & 0 & 0 & 0 & 1 & 0 & 0 & 0 & 0 & 0 & 0 & 0 \\
1 & 0 & 0 & 1 & 0 & 1 & 0 & 0 & 0 & 1 & 0 & 0 & 0 & 0 & 0 & 0 \\
1 & 0 & 0 & 1 & 0 & 0 & 1 & 0 & 0 & 0 & 1 & 0 & 0 & 0 & 0 & 0 \\
0 & 1 & 1 & 0 & 0 & 0 & 0 & 1 & 0 & 0 & 0 & 1 & 0 & 0 & 0 & 0 \\
1 & 0 & 0 & 0 & 0 & 1 & 1 & 0 & 0 & 0 & 0 & 0 & 1 & 0 & 0 & 0 \\
0 & 1 & 0 & 0 & 1 & 0 & 0 & 1 & 1 & 0 & 0 & 0 & 0 & 1 & 0 & 0 \\
0 & 0 & 1 & 0 & 1 & 0 & 0 & 1 & 1 & 0 & 0 & 0 & 0 & 0 & 1 & 0 \\
0 & 0 & 0 & 1 & 0 & 1 & 1 & 0 & 1 & 0 & 0 & 0 & 0 & 0 & 0 & 1 \\
1 & 0 & 0 & 0 & 0 & 0 & 0 & 0 & 0 & 1 & 1 & 0 & 1 & 0 & 0 & 0 \\
0 & 1 & 0 & 0 & 0 & 0 & 0 & 0 & 1 & 0 & 0 & 1 & 0 & 1 & 0 & 0 \\
0 & 0 & 1 & 0 & 0 & 0 & 0 & 0 & 1 & 0 & 0 & 1 & 0 & 0 & 1 & 0 \\
0 & 0 & 0 & 1 & 0 & 0 & 0 & 0 & 0 & 1 & 1 & 0 & 0 & 0 & 0 & 1 \\
0 & 0 & 0 & 0 & 1 & 0 & 0 & 0 & 1 & 0 & 0 & 0 & 0 & 0 & 0 & 0 \\
0 & 0 & 0 & 0 & 0 & 1 & 0 & 0 & 0 & 1 & 0 & 0 & 0 & 0 & 0 & 0 \\
0 & 0 & 0 & 0 & 0 & 0 & 1 & 0 & 0 & 0 & 1 & 0 & 0 & 0 & 0 & 0 \\
0 & 0 & 0 & 0 & 0 & 0 & 0 & 1 & 0 & 0 & 0 & 1 & 0 & 0 & 0 & 0 
\end{array}
\right).
$$

If $A_{[1]}$ is obtained from $A$ by setting all  elements in the first row equal to zero and $I_{16}$ is the $16 \times 16$ identity matrix then the  functions we need are the entries of the first column  of  $(I - z \tilde{A})^{-1}$.  They come in four sets corresponding to $L_1,L_2,L_3,L_4$, above, respectively:
$$\frac{(1-10z^2)z}{24z^4-16z^2+1},\;\frac{2(1-4z^2)z^2}{24z^4-16z^2+1},\;\frac{6z^3}{24z^4-16z^2+1},\;\frac{24z^2}{24z^4-16z^2+1}.
$$
If $\zeta(z)$ is one of these functions then using $k=4$ the distances are given by $\frac{1}{4} \zeta'(\frac{1}{4})$, and are respectively:
 $$15, \frac{56}{3}, \frac{61}{3}, \frac{64}{3}.$$

The following proof method can be made general but we again restrict it to the case $n=4$. By considering the sets $L_0, \dots, L_4$ as equivalence classes we can replace $A$ by  a  $ 5 \times 5$ (in general $(k+1) \times (k+1)$) matrix $B$ of the following form representing the transition between the $L_i$:
$$ L_0 \Leftrightarrow L_1 \Leftrightarrow L_2 \Leftrightarrow L_3 \Leftrightarrow L_4.$$
The matrix has an interesting form:
$$ B = \left[
\begin{array}{ccccc}
0 & 4 & 0& 0 & 0\\
1 & 0 & 3& 0 & 0\\
0 & 2 & 0& 2 & 0\\
0 & 0 & 3& 0 & 1\\
0 & 0 & 0& 4& 0\\
\end{array}
\right].
$$

If we carry out the same procedure as we used for $A$,  namely take $(I_5 - z \tilde{B})^{-1}$, where $\tilde{B}$ is $B$ with entries in the first row set to zero, we find  the same $\zeta(z)$ functions (now distinct) as obtained by using the full incidence matrix.

The general result uses the so-called ``group algebra'' which  quotients out by the equivalence relation, described above, replacing the set of generators in each $L_i$ by their sum. The group's action is then, essentially, on the whole equivalence class with a new (pseudo) adjacency matrix of the form $B$ above.  By careful study of the structure of such matrices, we are able to derive the formula for the distances from an element $L_t$, ($t=0,\ldots, k-1$) to $L_k$ to be:
\begin{equation}\label{e:abelian.formula}
d_{t,k} = \sum_{s=1}^{k-t} \left\{ {\binom{k-1}{s-1}}^{-1}\sum_{r=s}^k {\binom{k}{r}}\right\}.
\end{equation}
A simple check confirms that $d_{0,4}$, $d_{1,4}$, $d_{2,4}$, and $d_{3,4}$ give the distances $\frac{64}{3}, \frac{61}{3}, \frac{56}{3}, 15$ as calculated above.  

\section{Mason's rule and zeta functions }\label{a:mason.loops}
The original  Mason  rule~\cite{mason1953feedback} writes  the formula  for $\tilde{M}_{ij}(s)$ in terms of specifically defined paths and loops. For completeness we  present the essence of the original construction.

For a vertex $i$, a {\em bundle}, $B_i$ of $k(B_i)$ loops  is defined as a set of disjoint loops which do {\em not} pass through vertex $i$.
Define the {\em weight} of any collection, $C$ of edges
$$w(C)= \prod_{ (i,j) \in C} q(i,j),$$
where $q(i,j)$ is the transmittance of $(i,j)$. By Sylvester's rule for matrix inversion, the required moment generating function matrix has entries:
$$
\tilde{M}_{ij}(s)  = \left((I - Q_{[j]})^{-1}\right)_{ij} = \frac{\det\left((I-Q_{[j]})_{[ij]}\right)}{\det(I-Q_{[j]})} 
$$
where $\left((I - Q_{[j]})^{-1}\right)_{ij}$ indicates the $(i,j)$ entry of the matrix, and where the numerator of the right hand expression is  the cofactor of the $(j,i)$ element of $I - Q_{[j]}$ (this is the source of the term ``co-factor rule").

Properties of determinants give the denominator and numerator in the Mason rule as originally expressed:
\begin{eqnarray*}
	\det\left(I-Q_{[j]}(s)\right) & = & 1 +  \sum_{t}(-1)^t\sum_{B_j:\;k(B_j)=t}w(B_j) \\
	\det\left(\left(I-Q_{[j]}\right)_{[ij]}\right) & = &\sum_{R} w(R_{ij})\sum_{t}\left(1+ (-1)^t \sum_{\substack{B_j\;:\;k(B_j)=t,\\ B_j \cap R_{ij} = \emptyset,t\not\in B_j}}w(B_j)\right),
\end{eqnarray*}
where $R=\{R_{ij}\}$, and $R_{ij}$ is a direct (non self-intersecting) path from $i$ to $j$.

Recall that in  the case of a  Cayley graph $G$  we  start with a regular graph with incidence matrix $A$ in which
$p_{ij} = \frac{1}{k}$ and $m_{ij}(s) = m(s)$ and  use the generic notation $Q = zA$.

We feel that it is of independent interest that the quantity
$$\zeta_G(z) = \det(I-zA)^{-1}$$
is a type of zeta function. There are several different  zeta functions for graphs, and this special type arises in the theory of dynamical systems in which the edge $i \rightarrow j$ is referred to as a {\em  shift}. It is related to the  Bowen-Lanford theory~\cite{bowen1970zeta} as follows.
Under a suitable definition of a closed path and the condition that $A$ is aperiodic it can be shown that
$$\zeta_G(z) = \prod_{\tau} (1- z^{|\tau|}),$$
where $\tau$ is a simple circuit. 

However, in  our case the graphs defined by the matrices $A_{[i]}$ are absorbing and therefore are not captured by this formula. Even  for the full graph, $A$ can be periodic. For example in the case of $B_3$ (the hyperoctahedral group, or signed permutation group, on three letters):
$$\det(I-zA)^{-1} = (3z-1)(3z+1)(z-1)^3(z+1)^3(2z-1)^6(2z+1)^6,$$
showing eigenvalue multiplicities and hence periodicities.

Despite these multiplicity issues we can still give a description  of the quantities of interest in the style given in the  Mason rule. In the same way that zeta functions count circuits, it seems that they are a natural vehicle in this area to capture the loops and paths of the original theorem. Just as for several types of zeta function, determinants play a key role.

Thus, define
$$\zeta_G(z)^{-1} = 1 + \sum_j (-1)^j \sum_{B_i: \;k(B_i) = j} z^{k(B_i)}.$$
This suggests defining a vertex-specific zeta function for the Cayley graph $G$ by
$$\zeta_{G,i}(z)^{-1} = \det (I-z A_{[i]}) = 1+  \sum_j (-1)^j \sum_{B_i: \; k(B_i) = j; \; j \notin B_i} z^{k(B_i)}.$$
In a similar way, for an edge $(i,j)$ we have 
$$\zeta_{G,ij}(z)^{-1} =   \det ((I-z A_{[i]})_{[ij]}) = \sum_R z^{|R|} \sum _j \left(1+(-1)^j \sum_{B_i:\; k(B_i) = j, B_i \cap R = \emptyset; \; j \notin  B_i} z ^{|B_i|}\right),$$
where $R$ is defined as before.

\section{Discussion}
We collect here some questions, comments, and ideas for further investigation which arose in the gestation of this project, and of which some will be covered in our own future work.

\begin{enumerate}
\item It is clear that from a purely algebraic viewpoint the distance we propose houses information about the groups. This is very analogous to the way a zeta function holds information. One could say we have a special type of zeta function theory.
\item Since the Cayley graph depends not just on the group but the choice of generators for the group, so then does the distance. Thus the distances may be useful in separating out different biological processes by considering group generators corresponding to different sets of inversions, or other invertible operations such as translocations.
\item We should make clear that the distance is {\em linear} in the interarrival means $\mu_{ij}$. One could ask whether linearity is a useful property which may motivate further study.
\item That Cayley graphs are typically not aperiodic has been pointed out, and this is also clear from the circular structures in some examples. By adding additional generators they can be made aperiodic, and hence should make the steady state (ergodic) properties easier to study. 
\item We have made some simple assumptions about the interarrival moment generating function and the transitions, for our examples in Section~\ref{s:cayley}. But these can be made more general, for example by allocating different transition values to different types of biological event. An example of this may be a group-based model including both inversions and translocations, or one with different weights for inversions of different numbers of regions, as in~\cite{bhatia2019flexible}.
\item It is clear from examples, and the Abelian group example, that derivation of general formulae such as~\eqref{e:abelian.formula} may be achieved by a reduction using conjugacy, and a dummy Cayley graph with a pseudo-incidence structure such as in the matrix $B$ in Section~\ref{s:abelian}.
\item The mean first passage time (MFPT) distance provides a partial order on the elements of the group, much like the minimal distance, and others like the maximum likelihood estimate (MLE) distance (although this is only on a subset of the group).  The MLE distance has been shown to reverse the ordering on some group elements, relative to the minimal distance~\cite{serdoz2017maximum}, which has a concrete implication for phylogenetic reconstruction using algorithms like Neighbour-Joining~\cite{saitou1987neighbor}. It would be very interesting to understand whether any reversals of the minimal order under the MFPT are the same as those for the MLE.
\end{enumerate}

\bibliographystyle{plain}

\appendix
\section{Evolutionary models}
\label{s:evol.models}

Table~\ref{tab:models} shows a range of group-based models that this approach can be applied to.  Each corresponds to a particular group and generating set.

\begin{sidewaystable}

\caption{Models and corresponding group features.}\label{tab:models}

\begin{tabular}{|l|p{8cm}|p{13cm}|}
\hline
Name & Model & Group and Generators \\
\hline
\multicolumn{2}{|c|}{\emph{Unsigned inversions}} & The group is the symmetric group $S_n$, or Coxeter group of type $A_{n-1}$.\\
\hline
$U_2^{[L]}$ & 2-inversions on a lineal chromosome & Generators (allowable inversions) are the Coxeter generators $(1\ 2)$, $(2\ 3)$, \dots, $(n-1\ n)$.\\[2mm]
$U_2$ & 2-inversions on a circular chromosome. & Generators are the ``circular'' Coxeter generators  $(1\ 2)$, $(2\ 3)$, \dots, $(n-1\ n)$, $(n\ 1)$.  This was studied using the affine symmetric group in~\cite{egrinagy2013group}.\\[2mm]

$U_{3}$ & 2- and 3-inversions, circular chromosome. & Generators are $(i\ i+1)$ and $(i\ i+2)\mod n$.\\[2mm]
$U_{4}$ & 2-, 3- and 4-inversions, circular chromosome. & Generators are $(i\ i+1)$, $(i\ i+2)$, and $(i\ i+3)(i+1\ i+2)\mod n$.\\[2mm]
$U_k$ &  etc & \\[2mm]
$U_{full} $ & All inversions, circular chromosome.  & This is the dominant model in inversion distance literature.  \\[2mm]
\hline
\multicolumn{2}{|c|}{\emph{Signed inversions}} & The group is the hyperoctahedral group, or the Coxeter group of type $B_{n-1}$.\\
\hline
$S_{Cox}$ & & The Coxeter model with generators $t, s_1,\dots,s_{n-1}$, where $t$ swaps 1 and $-1$ and $s_i=(i\ i+1)$.  Not biologically sensible, but of interest algebraically.\\[2mm]
$S_2$ & Signed 1- and 2-inversions & Generators $(i\ -i)$, $(1\ -2)$ etc.  With signed inversions we assume $\pi(-i)=-\pi(i)$, so $(1\ -2)$ denotes $(1\ -2)(-1\ 2)$.\\[2mm]
$S_k$ & Signed inversions up to $k$ regions & \\[2mm]
\hline
\multicolumn{2}{|p{8cm}|}{\emph{Combined inversion and translocation model}} & For any of the above, include translocations that shift $i$ past $j$.  Write $t_{i,j}=(i\ i+1\ \dots\ j)$. \\
\hline

\hline
\end{tabular}

\end{sidewaystable}

\section{Cayley graphs for $S_4$}\label{s:cayley.figs}

The Cayley graphs of $S_4$ with standard and with circular generators are shown for reference in Figures~\ref{f:cayley.S4} and~\ref{f:cayley.S4.circ} respectively.

\begin{figure}[ht]
\includegraphics[width=12cm]{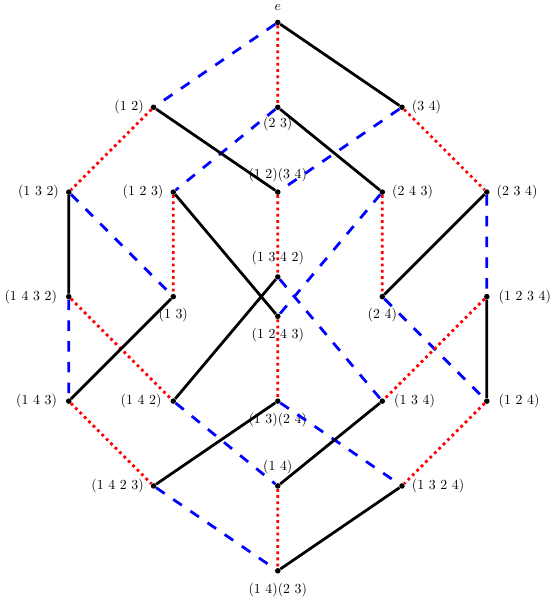} 
\caption{Cayley graph of $S_4$ with standard Coxeter generators. Edges are coded dashed blue for multiplication on the right by $(1\ 2)$, dotted red for $(2\ 3)$, and black for $(3\ 4)$ (group action is also on the right). }
\label{f:cayley.S4}
\end{figure}

\begin{figure}[ht]
\includegraphics[width=\textwidth]{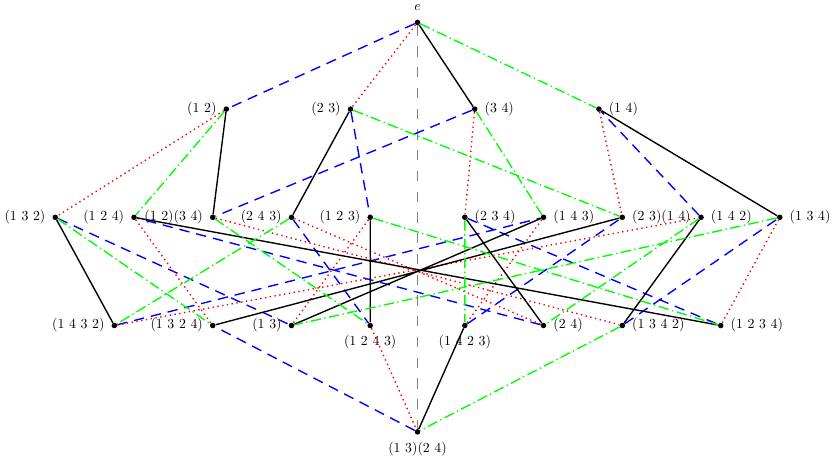}
\caption{Cayley graph of $S_4$ with circular generators. Edges are coded dashed blue for multiplication on the right by $(1\ 2)$, dotted red for $(2\ 3)$, black for $(3\ 4)$, and dash-dotted green for $(1\ 4)$ (group action is also on the right). 
}
\label{f:cayley.S4.circ}
\end{figure}

\end{document}